\documentclass[11pt]{article}

\usepackage{epstopdf}
\usepackage{amsmath}
\usepackage{amsfonts}
\usepackage{amssymb}
\usepackage{authblk}
\usepackage[boxed]{algorithm2e}
\usepackage{graphicx}
\usepackage{pstricks}
\usepackage{pst-plot}
\usepackage{pst-node}
\usepackage{cite}
\usepackage{color}
\usepackage{enumerate}
\usepackage{hyperref}
\usepackage[margin=1in]{geometry}
\usepackage[font=small,labelfont=bf]{caption}

\let\rOmega\Omega\renewcommand{\Omega}{\mathrm{\rOmega}}
\let\rTheta\Theta\renewcommand{\Theta}{\mathrm{\rTheta}}
\renewcommand{\O}{\mathcal{O}}

\newcommand{\M}{\mathrm{M}}
\newcommand{\F}{\mathrm{F}}

\newcommand{\N}{\mathrm{N}}
\renewcommand{\P}{\mathrm{P}}
\renewcommand{\l}[1]{\lvert #1\rvert}
\renewcommand{\epsilon}{\varepsilon}

\newtheorem{theorem}{Theorem}[section]

\newtheorem{lemma}[theorem]{Lemma}

\newtheorem{corollary}[theorem]{Corollary}

\newcommand{\qed}{\mbox{}\hspace*{\fill}\nolinebreak\mbox{\rule{6pt}{6pt}}}
\newenvironment{proof}{\vspace{-2mm}\noindent {\bf Proof:}}{\qed\par\medskip}

\begin{document}

\title{On the complexity of computing prime tables}

\author{Mart\'{\i}n Farach-Colton}
\author{Meng-Tsung Tsai}
\affil{Rutgers University \\
\vspace{0.2cm}
\{farach,mtsung.tsai\}@cs.rutgers.edu
}
\date{}

\maketitle

\begin{abstract}
Many large arithmetic computations rely on tables of all primes less than $n$. For example, the fastest algorithms for computing $n!$ takes time $\O(\M(n\log n) + \P(n))$, where $\M(n)$ is the time to multiply two $n$-bit numbers, and $\P(n)$ is the time to compute a prime table up to $n$.  The fastest algorithm to compute $\binom{n}{n/2}$ also uses a prime table.  We show that it takes time $\O(\M(n) + \P(n))$.  
\medskip

In various models, the best bound on $\P(n)$ is greater than $\M(n\log n)$, given advances in the complexity of multiplication~\cite{Furer07,De08}.  In this paper, we give two algorithms to computing prime tables and analyze their complexity on a multitape Turing machine, one of the standard models for analyzing such algorithms.  These two algorithms run in time $\O(\M(n\log n))$ and $\O(n\log^2 n/\log \log n)$, respectively.  We achieve our results by speeding up Atkin's sieve.
\medskip

Given that the current best bound on $\M(n)$ is $n\log n 2^{\O(\log^*n)}$, the second algorithm is faster and improves on the previous best algorithm by a factor of $\log^2\log n$.  Our fast prime-table algorithms speed up both the computation of $n!$ and $\binom{n}{n/2}$.
\medskip 

Finally, we show that computing the factorial takes $\Omega(\M(n \log^{4/7 - \epsilon} n))$ for any constant $\epsilon > 0$ assuming only multiplication is allowed. 
\end{abstract}

\smallskip
\textbf{Keywords.} prime tables, factorial, multiplication, lower bound

\section{Introduction}

Let $\P(n)$ be the time to compute prime table $T_n$, that is, a table of all primes from $2$ to $n$.  The best bound for $\P(n)$ on a log-RAM is $\O(n/\log\log n)$, using the Sieve of Atkin, and $\O(n\log^2 n\log\log n)$ on the multitape Turing machine (TM), a standard model for analyzing prime table computation, factorial computation, and other large arithmetic computations~\cite{Schonhage71,Schonhage94,Furer07}. This TM algorithm is due to Sch\"{o}nhage et al.~\cite{Schonhage94} as is based on the 
Sieve of Eratosthenes.  

The main result of this paper is two algorithms that improve the time to compute $T_n$ on a TM.
One runs in $\O(n\log^2 n/\log \log n)$ 
and thus speeds up Sch\"{o}nhage's algorithm by a factor of $\log^2\log n$.

The other has a running time that depends on the time to multiply large numbers.  Let $\M(a,b)$ be the time to multiply an $a$-bit number with a $b$-bit number, and let $\M(a) = \M(a,a)$. We make the standard assumption~\cite{Knuth97} that $f(n) = \M(n)/n$ is a monotone non-decreasing function. Then we give a prime-table algorithm that runs in time $\O(\M(n\log n))$ on a TM. F\"{u}rer's algorithm~\cite{Furer07} gives the best bound for $\M(n)$ on a TM, which is $n\log n 2^{\O(\log^* n)}$, a bound that was later achieved by a different method by De et al.~\cite{De08}, so our second algorithm is currently slower than the first algorithm.  

Prime tables are used to speed up many types of computation.  For example, the fastest algorithms for computing $n!$ depend on prime tables~\cite{Borwein85,Vardi91,Schonhage94}.
 Sch\"{o}nhage's algorithm~\cite{Schonhage94}  is fastest and takes time $\O(\M(n\log n) + \P(n))$.

The number of bits in $n!$ is $\Theta(n\log n)$, and Borwein~\cite{Borwein85} conjectured that computing $n!$ takes  $\Theta(\M(n\log n))$ time.  On the log-RAM, 
F\"{u}rer~\cite{Furer14} showed that $\M(n) = \O(n)$.  So on the log-RAM, the upper bound of Borwein's conjecture seems to be true, since $\M(n\log n)$ dominates $\O(n/\log\log n)$ for now.

On a TM, there is a simple lower bound of $\Omega(n\log n)$ to compute $n!$, since that is the number of TM characters needed to represent the output.  This contrasts with the $\O(n)$-word output on the log-RAM.  On the other hand, no $\O(\M(n\log n))$-time algorithm was known in this model, since before our improved prime-table algorithms, $\P(n)$ dominated $\M(n\log n)$\footnote{We note that before F\"{u}rer's algorithm, the opposite was true.  This is because before F\"{u}rer's algorithm, the best bound on $\M(n)$ was $\O(n\log n\log\log n)$~\cite{Schonhage71}.}.
Using our $\O(\M(n\log n))$-time prime-table algorithm, the time to compute $n!$ is improved to $\O(\M(n\log n))$.  If Borwein's conjecture turns out to be true, this algorithm will turn out to be optimal for computing $n!$.

Another use of prime tables is in the computation of binomial coefficients.
The exact complexity of computing binomial coefficients hadn't been analyzed, but here we show that a popular algorithm takes time $\O(\M(n) + \P(n))$.  Thus our faster algorithm also improves this running time by $\log^2\log n$.

Finally, we consider lower bounds for computing $n!$.  
Although we do not produce a general lower bound for computing $n!$ on a TM\footnote{And indeed, such a result would be a much bigger deal than any upper bound!}, we do show a lower bound for algorithms on the following restricted model. We do not restrict which operation can be used but we assume that the factorial $n!$ is output by a multiplication. We assume that a multiplication can only operate on two integers, each of which can be an integer of $o(n \log n)$ bits or a product computed by a multiplication. Under this restriction, we show a lower bound
\begin{equation}
\label{eqn:glower}
\Omega\bigg(\max_t\bigg\{\M_{t^{1/2-\epsilon}}\bigg(\frac{1}{t} n\log n\bigg), \frac{t}{w} n\log n\bigg\}\bigg) \mbox{ for } t \in [1, n], 
\end{equation}
where $w$ denotes the word size in the model.  Given an upper bound and a lower bound for $\M(n)$, we can simplify the lower bound in Equation~\eqref{eqn:glower}.

On the Turing Machine, we know that $\M(n)$ has a simple linear lower bound $\Omega(n)$ and, due to F\"{u}rer~\cite{Furer07} and De et al.~\cite{De08}, we have an upper bound $\M(n) = n \log n 2^{\O(\log^* n)}$. In that case, we have a lower bound in the multiplication model of
$$
\Omega(\M(n \log^{4/7 - \epsilon} n)) \mbox{ for any constant } \epsilon > 0.
$$
On the log-RAM, we know that $\M(n)$ has a lower bound of $\Omega(n/\log n)$ because operations on $\O(\log n)$ bit words take at least constant time. The upper bound for $\M(n)$, also due to F\"{u}rer~\cite{Furer14}, is $\O(n)$. In that case, under the multiplication restriction, we have the same lower bound as on the Turing Machine. They coincide because both models have a $\log^{1+\epsilon} n$ gap between the lower and upper bounds of $\M(n)$.

\paragraph{\textbf{Organization.}} In Section~\ref{sec:related}, we present the related work for computing prime tables. We propose two algorithm in Section~\ref{sec:prime}. Last, in Section~\ref{sec:lower}, we show a lower bound of computing factorials. The related work and new upper bounds for factorials and binomials can be found in the appendix, Sections~\ref{sec:background},~\ref{sec:upper}. 

\section{Background and Related Work \label{sec:related}}

In this section, we present the relevant background and related work on computing prime tables and defer those for factorials and binomial coefficients to Section~\ref{sec:background}.

The Sieve of Eratosthenes is the standard algorithm used in RAM model.  It creates a bit table where each prime is marked with a 1 and each composite is marked with a 0.  The multiples of each prime found so far are set to 0, each in $\O(1)$ time, and thus the whole algorithm takes time $\sum_{p \le n} n/p  = \O(n \log\log n)$. However, on a TM, each multiple of a prime cannot be marked in $\O(1)$ time. Instead, marking all the multiples of a single prime takes $\O(n)$ time, since the entire table must be traversed.  Since any composite number up to $n$ has some prime factor of at most $\sqrt{n}$, and there are $\O(\sqrt{n}/\log n)$ such primes, this approach takes time $\O(n^{3/2}/\log n)$.

Sch\"{o}hage et al. give an algorithm to compute a prime table from 2 to $n$ in $\O(n \log^2 n \log\log n)$ time~\cite{Schonhage94}. His algorithm, for each prime $p\leq\sqrt{n}$, generates a sorted list\footnote{It is not the case that each list occupy a tape; otherwise, $\omega(1)$ tapes are required. To merge these lists, put half of the lists on a tape, half on the other, merge them pairwise, output the sorted lists on another two tapes and recurse. In this way, 4 tapes are enough.} of the multiples of $p$, and then merges the $\O(\sqrt{n}/\log n)$ lists so generated. The total number of integers on these lists is $\O(n \log \log n)$, each integer needs to be merged $\O(\log n)$ times, and each integer has $\O(\log n)$ bits. Therefore, Sch\"{o}hage's algorithm has running time $\O(n \log^2 n \log\log n)$.

Alternatively, one can use the AKS primality test~\cite{Agrawal02} on each integer in the range from $2$ to $n$. The fastest known variant of the AKS primality test is 
due to Lenstra and Pomerance and takes $\tilde{\O}(\log^{6} n)$ time per test on a TM.  If Agrawal's conjecture~\cite{Agrawal02} is true, it takes $\tilde{\O}(\log^3 n)$ time. Whether the conjecture is true or not, it would still take $\Omega(n \log^3 n)$ time to compute a prime table.  One can use the base-2 Fermat test,
$$
2^n \equiv 2 \pmod{n},
$$
to screen out a majority of composite numbers. This would take $\O(n\log n\M(\log n))$, which is dominated by the AKS phase.
All prime numbers and $o(n/\log n)$ composite numbers can pass the base-2 Fermat test~\cite{Guy04}. Therefore, it reduces the complexity by a $\log n$ factor.  In this case, it would take a finer analysis of AKS  and settling Agrawal's conjecture to determine the exact complexity of this algorithm.  It would likely take 
$\tilde{\O}(n \log^2 n) = \O(n\log^2 n \log^k\log n)$ for some $k>0$, and this would improve on Sch\"{o}hage's algorithm if $k<1$.

We show how to implement the Sieve of Atkin to achieve a running time $\min\{\O(n \log^2 n/\log \log n), \O(\M(n \log n))\}$ on the Turing Machine in Section~\ref{sec:prime}.

\section{Fast algorithms for Atkin's Sieve \label{sec:prime}}

In this section, we give two algorithms for implementing Atkin's Sieve on a TM.  The first runs in time $\O(n\log^2 n/\log\log n)$.  The second runs in time $\O(\M(n\log n))$.  Given the state of the art in multiplication, the first is faster.  We present both, in case a faster multiplication algorithm is discovered.

\subsection{Atkin's Sieve in $\O(n\log^2 n/\log\log n)$}

We define some notions before proceeding to the proof. A \emph{squarefree} integer denotes an integer that has no divisor that is a square number other than 1. Let $\N_{f(x, y)}(k) = 0$ if there are even number of integer pairs $(x, y)$ that have $x > 0, y > 0$ and $f(x, y) = k$; or 1, otherwise. Similarly, let $\N'_{f(x, y)}(k) = 0$ if there are even number of integer pairs $(x, y)$ that have $x > y > 0$ and $f(x, y) = k$; or 1, otherwise. The key distinction is that the latter requires that $x > y$. In~\cite{Atkin04}, Atkin and Bernstein show how to test primality based $\N$ and $\N'$, as shown in Theorem~\ref{thm:atkin}.

\begin{theorem}[\hspace{-1pt}{\cite[Theorems~6.1-6.3]{Atkin04}}]\label{thm:atkin}
For every squarefree integer $k \in 1+4\mathbb{N}$, $k$ is prime iff $\N_{x^2+4y^2}(k) = 1$; for every squarefree integer $k \in 1+6\mathbb{N}$, $k$ is prime iff $\N_{x^2+3y^2}(k) = 1$; 
for every squarefree integer $k \in 11+12\mathbb{N}$, $k$ is prime iff $\N'_{3x^2-y^2}(k) = 1$. 
\end{theorem}

We show how to compute $\N_{x^2+4y^2}(k)$ for all $k \in [1, n]$ in $\O(n \log^2 n/\log \log n)$ time. First, for each $x \in [1, n^{1/2}]$, one can enumerate a short list of $x^2 + 4\cdot 1^2, x^2 + 4\cdot 2^2, \ldots, x^2+4 \cdot(n^{1/2})^2$. Clearly, each short list is already sorted. Then, we merge short lists pairwisely until a single sorted list is obtained; therefore, the running time is $\O(n\log^2 n)$ because there are $\O(n)$ integers, each of which has $\O(\log n)$ bits and is encountered $\O(\log n)$ times in the merge process. 

To speed up this process by a factor of $\log\log n$, noted
in~\cite{Atkin04}, Atkin and Bernstein show that the integers on these short lists are seldom coprime to the first $\log^{1/2} n$ primes. There are $\O(n/\log \log n)$ such integers in total. One can speed up this process by screening out the integers on these short lists that are not coprime to the first $\log^{1/2} n$ primes. This filter step can be completed in $\O(n \log^{1/2} n \M(\log n))$ time and the reduced short lists can be merged in the desired time. The same technique can be applied to $\N_{x^2+3y^2}(k)$ and $\N'_{3x^2-y^2}(k)$ for all $k \in [1, n]$. 

\begin{lemma}\label{lem:semiprime}
Computing $\N_{x^2+4y^2}(k)$, $\N_{x^2+3y^2}(k)$ and $\N'_{3x^2-y^2}(k)$ for all $k$ in $[1, n]$ takes $\O(n\log^2 n/\log \log n)$ time on the Turing Machine.
\end{lemma}

We computed the Atkin conditions but now we need to get rid of all non-squarefree numbers. Therefore, we show that generating all non-squarefree numbers requires $\O(n\log n)$ time in Lemma~\ref{lem:sfree}. Merging these three lists followed by screening out the list of non-squarefree numbers gives a prime table, as summarized in Theorem~\ref{thm:fast}.

\begin{lemma}\label{lem:sfree}
Generating a sorted list of all non-squarefree integers in the range $[1, n]$ takes $\O(n \log n)$
time on the Turing Machine.
\end{lemma}
\begin{proof}
We first generate the sorted list $L_1$ of all non-squarefree integers that has a divisor $p^2$ for some prime $p < \log n$.  We initialize an array of $n$ bits as zeros, for each prime $p < \log n$, we sequentially scan the entire array to mark all $mp^2$ for integer $m$ by counting down a counter from $p^2$ to $0$. Note that it requires amortized $\O(1)$ time to decrease down the counter by $1$ due to the frequency division principle~\cite{Berkovich00}. Since there are $\O(\log n/\log \log n)$ such primes, the running time of this step is $\O(n \log n/\log \log n)$. We then convert the array into the sorted list $L_1$ as required, which takes $\O(n \log n)$ time.  

Next, we generate a sorted list $L_2$ of all non-squarefree integers that has a divisor $p^2$ for some prime $p \ge \log n$. We generate a sorted short list for each such prime $p$, containing all the integers $mp^2 < n$ for some integer $m$. Then, we merge these sorted short lists. Note that there are $\sum_{p \ge \log n} n/p^2 = \O(n/\log n)$ integers on these short lists, each integer has $\O(\log n)$ bits, and each integer is encountered $\O(\log n)$ times in the merging process. The running time is thus $\O(n \log n)$. We are done by merging $L_1$ and $L_2$. 
\end{proof}

\begin{theorem}\label{thm:fast}
The prime table $T_n$ from $2$ to $n$ can be computed on the Turing Machine in time
$$\P(n) = \O(n\log^2n/\log\log n).$$
\end{theorem}

\subsection{Atkin's Sieve in $\O(\M(n \log n))$}

We show that sieve of Atkin can be realized in $\O(\M(n \log n))$ time on the Turing Machine. We apply multiplication to the computation of $N_{f(x, y)}(k)$ and $N'_{f(x, y)}(k)$ for all $k \in [1, n]$.  
The balance of the work will take $\O(n\log n)$, and will thus be dominated by the multiplication.

An important aspect of the multiplication will be the number of bits needed in the multiplicands.  For this, we need Lemma~\ref{lem:mass}, stating an upper bound of the number of (integer) lattice points on the ellipses specified by the first two Atkin conditions and on the truncated hyperbola $3x^2 - y^2 = k$ for $x > y > 0$. 

\begin{lemma}\label{lem:mass}
The number of integer pairs $(x, y)$ that satisfy $x^2 + 4y^2 = k$ for any positive integer $k$ coprime to $6$ is bounded by $k^{\O(1/\log\log k)}$. The same bound holds for $x^2+3y^2 = k$ and $3x^2-y^2 = k, x > y > 0$.
\end{lemma}

\begin{proof}
Observe that every pair $(x, y)$ that satisfies $x^2 + 4y^2 = k$ induces an unique pair $(x' = x, y' = 2y)$ that satisfies ${x'}^2 + {y'}^2 = k$. Therefore, the number of pairs $(x, y)$ that satisfies the latter equation is no less than that of the former. It is known that, for any odd integer $k$, there are 
\begin{equation} \label{eqn:circle}
\O\left(\sum_{d \mid k} (-1)^{(d-1)/2}\right)
\end{equation}
integer pairs $(x', y')$ that satisfy ${x'}^2+{y'}^2 = k$~\cite{Grosswald85}. Since the number of divisors of an integer $k$ is no more than
$
\O\left(k^{1/\log \log k}\right)
$
due to Wigert~\cite{Dickson05}, an upper bound for~\eqref{eqn:circle} is $\O(k^{1/\log \log k})$. Similarly, it is known that for any odd integer $k$ there are
\begin{equation} \label{eqn:ellipse}
\O\left(\sum_{d \mid k} \left(\dfrac{-3}{d}\right) \right) 
\end{equation}
integer pairs $(x, y)$ that satisfy $x^2 + 3y^2 = k$~\cite{Heaslet39}, where $\left(\frac{a}{b}\right)$ denotes the Jacobi symbol. Because each Jacobi symbol has value no more than 1, an upper bound for~\eqref{eqn:ellipse} is $\O(k^{1/\log \log k})$ as desired.

We argue that, for any integer $k$ coprime to $6$, the number of integer pairs $(x, y)$ that satisfy equation $3x^2 - y^2 = k, x > y > 0$ has the same bound. We first give a proof for the case that $x, y, k$ are mutually relatively primes and then relax the restriction. 

Let $k = p_1^{r_1} p_2^{r_2} \cdots p_t^{r_t}$ where the $p_i$'s are distinct primes more than 3 and the $r_i$'s are positive integers. Observe that every integer pair $(x, y)$ that satisfy $3x^2 - y^2 = k, x > y > 0$ has the property that $x, y < k^{1/2}$. Therefore, every integer pair $(x, y)$ that satisfy $3x^2 - y^2 = k, x > y > 0$ induces an unique pair $(x' \equiv x \bmod k, y' \equiv x \bmod k)$ that satisfies $3x'^2-y'^2 \equiv 0 \pmod{k}$ as well as induces a pair $(x' \equiv x \bmod p_i^{r_i}, y' \equiv y \bmod p_i^{r_i})$ that satisfies $3x'^2 - y'^2 \equiv 0 \pmod{p_i^{r_i}}$.

We claim that any integer pair $(x, y)$ that satisfies $3x^2 - y^2 \equiv 0 \pmod{k}$ has an unique product $(yx^{-1} \bmod{k})$, where the inverse $x^{-1}$ exists since $x$ and $k$ are relatively prime. We give a proof by contradiction. Suppose $(x_1, y_1)$ and $(x_2, y_2)$ yield the same product $(yx^{-1} \bmod{k})$, then $y_1x_2 \equiv y_2x_1 \pmod{k}$ or, equivalently, $y_1x_2 = y_2x_1$ due to $x_1, y_1, x_2, y_2 < k^{1/2}$. Since $x_1$ and $y_1$ are relatively prime, and $x_2$ and  $y_2$ are relatively prime, then $x_1 = x_2$, $y_1 = y_2$, a contradiction. 

We show that the number of distinct products $(yx^{-1} \bmod{k})$ is at most $2^t$. Since $(x' \equiv x \bmod{p_i^{r_i}}, y' \equiv y \bmod{p_i^{r_i}})$ satisfies $3x'^2 - y'^2 \equiv 0 \pmod{p_i^{r_i}}$, $(a_i \equiv y'x'^{-1} \bmod{p_i^{r_i}})$ is a square root of $3$ modulo $p_i^{r_i}$. There are at most two distinct square roots of $3$ for each modulo $p_i^{r_i}$, $p_i > 3$~\cite[Theorem~5.2]{LeVeque02}. By the Chinese Remainder Theorem, $(a_1, a_2, \ldots, a_t)$ is in a one-to-one correspondence to $(yx^{-1} \bmod{k})$. Hence, there are at most $2^t$ distinct products $(yx^{-1} \bmod{k})$ as desired.

Consequently, the number of integer pairs $(x, y)$ that satisfy $3x^2 - y^2 = k, x > y > 0$ for any integer $k$ coprime to 6 is bounded by 
$$
\O\left(k^{1/\log\log k}\right) \mbox{ for } x, y, k \mbox{ are relatively primes.}
$$
For the case that two of $x, y, k$ have common divisor $d > 1$, then the third one also has the divisor $d$. Then, one can divide $x, y, k$ by the common divisor $d$, thus reducing to a case of $x, y, k'$ being mutually relatively prime for $k' < k$. There are $\O(k^{1/\log \log k})$ such smaller $k'$ and each smaller $k'$ contributes $\O(k^{1/\log \log k})$ pairs $(x, y)$ at most. We are done.
\end{proof}

\begin{lemma}\label{lem:ellipse}
Given a function $f(x, y) = ax^2 + by^2$ for $a > 0, b > 0$, $\N_{f(x, y)}(k)$ for all $k \in [1, n]$ can be computed in $\O(\M(n \log n))$ time.
\end{lemma}
\begin{proof}
Any positive integer pair $(x, y)$ that satisfies $f(x, y) = k$ has the property that $ax^2, by^2 < k$. We claim that a long multiplication on a pair of $\O(n \log n)$-bit integers suffices to compute $N_{f(x, y)}(k)$ for all $k \in [1, n]$.

For $i \in [1, n]$, let $\alpha_i = 1$ if some $ax^2 = i$, or otherwise $\alpha_i = 0$. Similarly, for $j \in [1, n]$, let $\beta_j = 1$ if some $by^2 = j$, or otherwise $\beta_j = 0$. Then, the following product of polynomials
\begin{equation*}
\sum_{i = [1, n]} \alpha_i z^i \sum_{j \in [1, n]} \beta_j z^j
\end{equation*}
has the property that the coefficient of $z^k$ modulo $2$ is equal to $\N_{f(x, y)}(k)$. One can use a multiplication to replace the product of polynomials by replacing $z$ with an integer base $B$. To avoid carry issue, we choose $B = \Theta(\log n)$ because the coefficient of $z^k$ is at least bounded by $\O(n^2)$. Thus, the running time is $\O(\M(n\log n)).$  
\end{proof}

\begin{corollary}\label{cor:ellipse}
Given functions $f(x, y) = x^2 + 4y^2, g(x, y) = x^2 + 3y^2$, $\N_{f(x, y)}(k)$ and $\N_{g(x, y)}(k)$ for all $k \in [1, n]$ can be computed in $\O(\M(n\log n/\log \log n))$ time. 
\end{corollary}
\begin{proof}
We use the algorithm stated in Lemma~\ref{lem:ellipse} but, due to Lemma~\ref{lem:mass}, we can choose $B$ to be $\Theta(\log n/\log \log n)$ rather than $\Theta(\log n)$. One needs to avoid the computation of $\N_{f(x, y)}(k)$ for $k$ not coprime to $6$ because $\N_{f(x, y)}(k)$ might require more than $\Theta(\log n/\log \log n)$ bits for such $k$. We avoid the computation of $\N_{f(x, y)}(k)$ for such $k$ by classifying $x^2$, $4y^2$, $3y^2$ into groups according to their residue modulo $6$. Then, multiplying these groups in pairs only if their sum is coprime to 6, which amplifies the complexity by a constant factor. 
\end{proof}

\begin{lemma}\label{lem:hyperbola}
Given a function $f(x, y) = 3x^2 - y^2$, $\N'_{f(x, y)} (k)$ for all $k \in [1, n]$ can be computed in $\O(\M(n \log n))$ time. 
\end{lemma}
\begin{proof}
Any positive integer pair $(x, y)$ that satisfies $f(x, y) = k$ and $x > y$ has the property that $x, y < k^{1/2}$. We claim that $\log n$ multiplications suffice to compute $N'_{f(x, y)}(k)$ for all $k \in [1, n]$. 

We relax the condition $x > y$ by divide and conquer and then process each subproblem as Lemma~\ref{lem:ellipse}. 
We reduce the range of pairs $(x, y)$, $0 < y < x < n^{1/2}$ to following three cases, let $h = n^{1/2}/2$:
(1) $x \in [h, n^{1/2}]$ and $y \in [0, h)$,
(2) $0 < y < x < h$,
(3) $h \le y < x < n^{1/2}$.

Note that case (1) can be computed by the product of $n$-term polynomial as what was done in Corollary~\ref{cor:ellipse} due to Lemma~\ref{lem:mass}. Therefore, case (1) can be done in $\O(\M(n \log n/\log \log n))$ time. Besides, the number of pairs $(x, y)$ in cases (2) and (3) is half of that in the original case. To match the claimed complexity, we recurse for $\log \log n$ levels, with a running time of $\O(\M(n \log n))$
and generate $\O(\log n)$ lists of pairs $(x, y)$ sorted in ascending $f(x, y)$ and we use the first algorithm in Lemma~\ref{lem:sfree} to merge them into a sorted list $L_1$ in $\O(n \log n)$ time. Note that, by the first algorithm, any pair of duplicated integers is discarded, since we only care about parity. After the recursion, the number of unprocessed pairs $(x, y)$ is $\O(n/\log n)$. We merge the unprocessed pairs $(x, y)$ into a single sorted list $L_2$ in ascending $f(x, y)$ by the second algorithm used in Lemma~\ref{lem:sfree}, which takes $\O(n \log n)$ time. Finally, we are done by merging $L_1$ and $L_2$. 
\end{proof}

Combining Lemma~\ref{lem:sfree},~\ref{lem:hyperbola} and Corollary~\ref{cor:ellipse}, we can realize the sieve of Atkin with a few of long multiplications and some minor procedures doable in $\O(n \log n)$ time. As a result, we have Theorem~\ref{thm:main}. 

\begin{theorem}\label{thm:main}
The prime table $T_n$ from $2$ to $n$ can be computed on the Turing Machine in time
$$\P(n) = \O(\M(n \log n)).$$
\end{theorem}

\section{Lower Bound \label{sec:lower}}

We present a lower bound for computing the factorial $n!$. 
We do not restrict which operation can be used but we assume that the factorial $n!$ is output by a multiplication. We assume that a multiplication can only operate on two integers, each of which can be an integer of $o(n \log n)$ bits or a product computed by a multiplication.
Under this assumption, we show that computing the factorial $n!$ has a lower
bound $\Omega(\M(n \log^{4/7-\epsilon} n))$ for any constant $\epsilon
> 0$. 

To show the claimed lower bound, we need some lemmas for $\M(n)$ and
$\M_k(n)$, where $\M_k(n)$ denotes the optimal time to multiply $k$
pairs of two $n$-bit integers. There is a subtle difference between
$\M_k(n)$ and $k \M(n)$. $\M_k(n)$ denotes the optimal time to
multiply $k$ pairs of integers, possibly in parallel, because all
these integers are given at the beginning; however, $k\M(n)$ denotes
the optimal time to multiply $k$ pairs of integers serially, one
after another. Hence, $\M_k(n) \le
k\M(n)$. Lemmas~\ref{lem:simple},~\ref{lem:onemany} are simple facts about
the Turing Machine model. Lemma~\ref{lem:manyone} is based on the
property of progression-free
set~\cite{Erdos36,Behrend46,Roth53,Bourgain07,Elkin10}.

\begin{lemma}
\label{lem:simple}
$\M(a, b) = \Omega(a+b) \mbox{ and } \M(a, b) = \Omega(\M(a)) \mbox{ if } a \le b.$
\end{lemma}

\begin{proof}
$\M(a, b) = \Omega(a+b)$ clearly holds on the Turing Machine model. 
To compute the product of two $a$-bit integers, every bits of the integers has to be read.
On a Turing Machine, one can read one character in a step. 
Since the alphabet set has constant size, every character can encode $\O(1)$ bits. 

We prove $\M(a, b) = \Omega(\M(a)) \mbox{ if } a \le b$ by padding zeros. Suppose $\M(a, b) = o(\M(a))$, then $b = o(\M(a))$. To multiply two $a$-bit integers, one can pad $b-a = o(\M(a))$ zeros to one $a$-bit integer and then multiply. In this way, the total running time is $\M(a, b) + o(\M(a)) = o(\M(a))$, contradicting the optimality of $\M(a)$. 
\end{proof}

\begin{lemma}
\label{lem:manyone}
The products of independent short multiplications can be computed by a long multiplication; in particular, 
$$\M_\ell (a) = \O(\M(k a)), $$
where $\log k^2 < a$ and $\ell = k^{1-\epsilon} \mbox{ for any constant } \epsilon > 0.$
\end{lemma}
\begin{proof}
We represent a $ka$-bit integer $A$ with a sum of $k$ terms
$$
A = \alpha_0 + \alpha_1 x^1 + \alpha_2 x^2 + \cdots + \alpha_{k-1} x^{k-1} \mbox{ where the base } x = 2^a 
$$
and likewise for $B$. We initialize $\alpha_i$'s and $\beta$'s with
zeros. For each short multiplication $u \times v$, we assign $\alpha_i
= u$ and $\beta_i = v$ for some index $i$, preserving the following
condition. We require that the set $S$ of assigned indices be
progression-free; that is, for every $i, j, h \in S$, $j+j \ne 2h$. In
this way, if we do the multiplication $C = AB$, then the product of
matched $u, v$ is placed at the coefficient of $x^{2i}$ and the
products of mismatched pair cannot be placed at $x^{2i}$ for any
$i$. However, carries can violate the claim. One can avoid a over-long
carry by not
assigning even numbers for indices or not assigning odd  numbers for
indices because $\log k^2 < a$. 

Every progression-free set $S \subseteq \{1, 2, \ldots, k\}$ has size at
most $k^{1-\epsilon}$ for any constant $\epsilon >
0$~\cite{Roth53,Bourgain07} and there exist efficient
algorithms for finding one set of that
size~\cite{Salem42,Behrend46,Moser53,Elkin10}. On a TM, one can use
Behrend's algorithm~\cite{Behrend46}, which relies on
finding a hyperball containing sufficiently many lattice points on
it, to find such a set. This can be reduced to multiplications as does
Lemma~\ref{lem:ellipse}.
By the Pigeon-hole principle, at least half the integers in $S$ are 
even or odd. Therefore, we can multiply $\ell$ pairs of two
$a$-bit integers by computing the product of two $ka$-bit
integers. 
\end{proof}

\begin{lemma}
\label{lem:onemany}
The products of a long multiplication can be computed by the products of independent short multiplications; in particular, 
$$\M(k^{1/2} n) = \O(\M_k (n)).$$
\end{lemma}

\begin{proof}
We partition the $k^{1/2}n$-bit integers into $k^{1/2}$ chunks. Then, to compute the product of the integers, we compute the products of pairwise chunks and then sum the products up. There are $k$ pairs of chunks and they have no dependency. That means the product of pairwise chunks can be computed in parallel, completing the proof. 
\end{proof}

Since we restrict that the factorial $n!$ is output by a multiplication, there must be a multiplication  
$a_1 \times b_1 = a_0 = n!$
in every algorithm. Besides,
we restrict that only the integers of $o(n \log n)$ bits and intermediate products can be multiplied. Therefore, $a_1, b_1$ are small integers or the computed intermediate
products. Let $\l{x}$ denote the number of bits in $x$.

If $\l{a_i} > \l{a_0}/2$, then $a_i$ has more than $o(n \log n)$ bits. Therefore, $a_i$ is also an intermediate product and assert the existence of a multiplication $a_{i+1} \times b_{i+1} = a_i$. We can repeat this until some $\l{a_i} \le \l{a_0}/2$. We define $t$ to be the step where it stops.
 Therefore, there must be $t$
multiplications, 
$
a_i \times b_i = a_{i-1} \mbox{ for all } i \in [1, t],
$
in any algorithm that can compute the factorial. In other words, we
have a lower bound of
\begin{equation}
\label{eqn:lower}
\sum_{i \in [1, t]} \M(\l{a_i}, \l{b_i}).
\end{equation}
W.l.o.g., let $\l{a_i} \ge \l{b_i}$ and therefore $\l{a_i} \ge \l{a_0}/4 \mbox{ for all } i \in [1, t].$

Let us simplify Equation~\eqref{eqn:lower} by observing the distribution of $b_i$'s. Consider that  
$$
a_t \prod_{i \in [1, t]} b_i = a_0 \mbox{ and } \sum_{i \in [1, t]} \l{b_i} \ge \l{a_0} - \l{a_t},
$$
then
$
\mu = (\l{b_1}+\l{b_2}+\cdots+\l{b_t})/t \ge \l{a_0}/(2t).
$
Furthermore, for any $\gamma \in [1, t]$, if there is no $b_i$ more than $\gamma\mu$, then there are $t/\gamma$ $b_i$'s more than $\mu/2$, which is an extension of Markov's inequality. We are ready to show the lower bound in Lemma~\ref{lem:tradeoff}.

\begin{lemma}
\label{lem:tradeoff}
Computing the factorial $n!$ has a lower bound
$$
\Omega\bigg(\M_{t^{1/2-\epsilon}}\bigg(\frac{1}{t} n\log n\bigg)\bigg)
$$
where $t$ is a parameter to be determined later. 
\end{lemma}

\begin{proof}
By applying the extended Markov inequality to Equation~\eqref{eqn:lower}, one has the lower bound
$$
\sum_{i \in [1, t]} \M(\l{a_i}, \l{b_i}) \ge \max_{\gamma \in [1, t]} \min\bigg\{\M(\l{a_0}/4, \gamma\mu),\frac{t}{2\gamma}\M(\l{a_0}/4, \mu/2)\bigg\},
$$
which is, by Lemma~\ref{lem:simple}, more than 
$$
\max_{\gamma\in [1, t]}\min\bigg\{\M\bigg(\frac{\gamma}{t} n \log n\bigg), \frac{t}{2\gamma}\M\bigg(\frac{1}{2t} n \log n\bigg)\bigg\}.
$$
We convert the two terms to the same form and compare. 
We apply Lemma~\ref{lem:manyone} for 
the first term and the mentioned $\M_k(n) \le k\M(n)$ bound for the
second term, thus obtaining
$$
\max_{\gamma\in [1, t]}\min\bigg\{\M_{(2\gamma)^{1-\epsilon}}\bigg(\frac{1}{2t} n \log n\bigg), \M_{\frac{t}{2\gamma}}\bigg(\frac{1}{2t} n \log n\bigg)\bigg\}
$$
for any constant $\epsilon > 0$. Observe that $\M_{k} (a) \le
\M_{\ell} (a)$ if $k \le \ell$. As a result, we
have the following lower bound, by choosing $\gamma =
t^{1/2+\epsilon}/2$ for any constant $\epsilon > 0$,  
\begin{equation}
\Omega\bigg(\M_{t^{1/2-\epsilon}}\bigg(\frac{1}{t} n\log n\bigg)\bigg).
\end{equation}
\end{proof}

Observe that Lemma~\ref{lem:tradeoff} yields a good lower bound only
if $t$ is small. Our strategy is to find another lower bound which is
good when $t$ is large. Then, we can trade off between these lower
bounds. We finalize the proof for the claimed lower bound in
Theorem~\ref{thm:lower}.  

\begin{theorem}
\label{thm:lower}
On a TM, computing the factorial $n!$ has a lower bound 
$$\Omega(n \log^{4/7-\epsilon} n) \mbox{ for any constant }\epsilon > 0.$$
\end{theorem}

\begin{proof}
By Lemma~\ref{lem:simple}, one has
$$
\sum_{i \in [1, t]} \M(\l{a_i}, \l{b_i}) \ge \sum_{i \in [1, t]} (\l{a_i} + \l{b_i}) \ge t\frac{\l{a_0}}{4}.
$$
Combining the above lower bound and the lower bound shown in Lemma~\ref{lem:tradeoff}, we obtain
\begin{equation}\label{eqn:reqlower}
\Omega\bigg(\min_{t}  \max\bigg\{\M_{t^{1/2-\epsilon}}\bigg(\frac{1}{t} n\log n\bigg), tn\log n\bigg\}\bigg).
\end{equation}
Again, we convert the two terms to the same form and compare. 
We apply Lemma~\ref{lem:onemany} for the first term and apply
the current upper bound of $\M(n) \le n \log n 2^{\O(\log^* n)}$ for
the second term. Then, the lower bound becomes 
\begin{equation}\label{eqn:requpper}
\Omega\bigg(\min_t \max \M\bigg(\frac{n \log n}{t^{3/4+\epsilon}}\bigg), \M\bigg(\frac{tn}{2^{\O(\log^* n)}}\bigg) \bigg).
\end{equation}

The optimal bound appears at $t = \log^{4/7-\epsilon} n$ for any constant $\epsilon > 0$ as desired.
\end{proof}

\begin{corollary} \label{cor:lower}
On a log-RAM, computing the factorial $n!$ has a lower bound 
$$\Omega(n \log^{4/7-\epsilon} n) \mbox{ for any constant } \epsilon > 0.$$
\end{corollary}
\begin{proof}
We replace the lower bound of $\M(n)$ in Equation~\ref{eqn:reqlower} with $\Omega(n/\log n)$ and replace the upper bound of $\M(n)$ in Equation~\ref{eqn:requpper} with $\O(n)$. By similar analysis, we are done. 
\end{proof}


\section{Background and Related Work (Cont'd) \label{sec:background}}

\subsection{Computing Binomial Coefficients}
Consider, as an example, the computation of the central binomial coefficient $\binom{n}{n/2}$, a simple algorithm for which is to compute $n!$ and $(n/2)!$ independently and divide $n!$ by the square of $(n/2)!$. However, $n!$ and $(n/2)!$ each have $\Theta(n \log n)$ bits, which is much more than the $\Theta(n)$ bits that $\binom{n}{n/2}$ has. One can do something clever by cancelling the common factors between the numerator and denominator.  For example, when $n$ a multiple of 24,  
$$
\binom{n}{n/2} = 
Q_{n} Q_{n/2}^{-1}Q_{n/4}^{-1} \binom{n/4}{n/8} \binom{n/3}{n/6} \binom{n/12}{n/24}^{-1},
$$
where $Q_{n}$ is the product of positive integers in $\{k \le n \mid gcd(k, 6) = 1\}$. For $n$ not a multiple of 24, there are at most $23 = \O(1)$ further multiplications needed to compute the value. 

This approach reduces the number of multiplications from $n$ to $19 n/24$. Some of these multiplications can be further reduced by a recursive call; however, $Q_n$ requires $\Omega(n)$ multiplications. 

One can reduce the number of multiplication required for $Q_n$ by letting $Q_n$ be the product of integers in $\{k \le n \mid gcd(k, p) = 1 \mbox{ for each prime } p \le t\}$, where $t$ is a chosen threshold. By Merten theorem~\cite{Dickson05}, $Q_n$ is a product of $\O(n / \log t)$ integers. To make $Q_n$ be a product of $\O(n /\log n)$ integers, it is necessary to sieve out the multiples of $n^{\Omega(1)}$ primes.  The running time for this 
matches that for computing prime tables if Sch\"{o}nhage algorithm is used.

Suppose a smaller $t$ is chosen, this approach needs $\Omega((\log n)\M(n \log n/\log t))$ time, which is more than $\M(n\log n)$, assuming that the conjecture $\M(n) = \Theta(n \log n)$ holds.

A folk algorithm to compute the binomial coefficient $\binom{n}{k}$ more efficiently given $T_n$  is based on Kummer's theorem~\cite{Kummer52}, stating that for each prime $p$, the largest natural number $r$ such that $p^r$ divides $\binom{n}{k}$ can be computed in $\O(\log n/\log p)$ time. We analyze the complexity of this approach and show in Section~\ref{sec:upper} that it is $\O(\M(n)+\P(n))$.

\subsection{Computing $n!$}

There exist several efficient algorithms to calculate $n!$~\cite{Borwein85,Boiten91,Vardi91,Schonhage94,Ugur06}. Some~\cite{Borwein85,Boiten91} focus on reducing the total number of bits of intermediate products by grouping the $n$ integers into sub-groups, for example by commuting the product of each pair of successive integers, and then each pair of those products,  and so on.  The total number of bits of intermediate products is then greatly reduced to $\O(n \log^2 n)$.

Others~\cite{Borwein85,Vardi91,Schonhage94,Ugur06} focus on reducing the amount of shared computation between multiplications.  The idea is to use the observation that $p^n$ can be computed via $\O(\log n)$ multiplications, with intermediate products $p^2, p^4, p^8, \ldots, p^n$, instead of by $\O(n)$ iterative multiplications by $p$.  In order to use this to compute $n!$, such algorithms decompose $n!$ into prime factors, say $p_1^{r_1}p_2^{r_2}\cdots$,  and achieve their speedups by carefully scheduling multiplications in order to reduce the number of intermediate products.  

Borwein~\cite{Borwein85} divides the factors $p_i$ into $\O(\log n)$ groups $G_1, G_2, \ldots$ where 
\begin{equation*}
G_j = \{p_i \mid \mbox{the $j$-th bits of $r_i$ is 1 in base 2}\}.
\end{equation*}
Let $\pi_{j} = \prod_{p \in G_j} p$. Since each factor in the same group $G_j$ has the same exponent, then $n! = \prod_{j} \pi_{j}^{2^{j-1}}$. One can compute the product $\pi_j$ first and compute its power $\pi_j^{2^{j-1}}$ later. This greatly reduces the amount of shared computation. Borwein shows that this approach runs in $\O(\M(n\log n)\log \log n + \P(n))$ time.  Note that, as Sch\"{o}nhage pointed out, Borwein did not include the time to compute the prime table but took the table as given.

Sch\"{o}nhage et al.~\cite{Schonhage94} presented a variation of Borwein's algorithm by factoring $n!$ as follows:
\begin{equation}
\label{eqn:Schonhage}
n! = \pi_{1} (\pi_2 (\pi_3 (\pi_4 \cdots )^2)^2 )^2.
\end{equation}
This approach takes advantages on the fact that multiplying before exponentiating is faster than exponentiating each term in a product independently. This algorithm  has run time 
$\O(\M(n \log n) + \P(n))$.  Sch\"{o}nhage gave an $\O(n\log^2n\log\log n)$ algorithm to compute a prime table.  At the time of publication, this constituted a $\log\log n$ factor improvement over Borwein's algorithm for computing $n!$.  Given F\"{u}rer's improvement on multiplication, this improvement is down to $2^{\O(\log^* n)}$.

Using an approach similar to Sch\"{o}nhage's, Vardi~\cite{Vardi91} independently gave an algorithm based on the identity:
\begin{equation}
\label{eqn:Vardi}
n! = \binom{n}{n/2} \bigg(\binom{n/2}{n/4} \bigg(\binom{n/4}{n/8} \bigg(\binom{n/8}{n/16} \cdots\bigg)^2\bigg)^2\bigg)^2 \mbox{ for } n = 2^k.
\end{equation}
One might wonder what the difference is between Equations~\eqref{eqn:Schonhage} and~\eqref{eqn:Vardi} at the first glance. Note that $\pi_1 = \binom{n}{n/2}$
if and only if $\binom{n}{n/2}$ is squarefree and similar to other $\pi_i$'s.
However, Erd\"{o}s' squarefree conjecture~\cite{Erdos80} states that $\binom{2n}{n}$ is never squarefree for $n > 4$.  This was proved by Granville and Ramar\'{e}~\cite{Granville96}.  This implies that Sch\"{o}nhage's algorithm performs fewer  multiplications than Vardi's. Vardi did not analyze the complexity his algorithm. We analyze Vardi's algorithm in Section~\ref{sec:upper} and show that it has the same asymptotic complexity as Sch\"{o}nhage's, that is $\O(\M(n \log n)+\P(n))$, as long as the binomial coefficients are computed in time $\O(\M(n)+\P(n))$. 

However, it is possible that a faster algorithm to compute binomial coefficients exists, one that does not rely on prime table computation.  Therefore, there exists some hope that the second term in the time complexity might be removed, even if no faster algorithm is given for prime table computation.

\section{Factorials and Binomials \label{sec:upper}}

We analyze the complexity of computing the factorial $n!$ by Vardi's algorithm~\cite{Vardi91}. Since Vardi's algorithm relies on the computation of central binomial coefficients, we begin by analyzing the complexity of computing the binomial coefficient $\binom{n}{k}$.

\subsection{Computing Binomial Coefficients in $\O(\M(n)+\P(n))$ Time}

It is known that binomial coefficients can be efficiently computed by Kummer's Theorem~\cite{Kummer52,Vardi91}. However, the exact complexity is not known. Here we give an analysis.  

Kummer's Theorem~\cite{Kummer52} states that, for any binomial
coefficient $\binom{n}{k}$, any prime $p$, the maximum integer $r$
such that $p^r$ divides $\binom{n}{k}$ is equal to the number of
carries occur when adding $n-k$ and $k$ in base $p$. Therefore, the
prime factorization $p_1^{r_1}p_2^{r_2}\cdots p_t^{r_t}$ of
$\binom{n}{k}$ can be computed by trying every possible prime from
$2$ to $n$. Each trial requires $\O((\log_p n)\M(\log n))$ time
because division and modular arithmetics on $\O(\log n)$-bits integers require $\O(\M(\log n))$ time~\cite{Knuth97}. Hence, the prime factorization of $\binom{n}{k}$ can be obtained in $\O(\M(n))$ time due to Lemma~\ref{lem:sum}.

\begin{lemma}\label{lem:sum}
Let $p_1^{r_1}p_2^{r_2}\cdots p_t^{r_t}$ be the prime factorization of $\binom{n}{k}$. Then, 
$$\sum_{i \le t} r_i = \O\left(\sum_{\mbox{\scriptsize prime } p \le n} \log_{p} n\right) = \O(n/\log n).$$
\end{lemma}
\begin{proof}
By Kummer's Theorem~\cite{Kummer52}, we have $r_i = \O(\log_{p_i} n)$. Since 
$$
\sum_{\mbox{\scriptsize prime } p \le n} \log_p n \le 
\sum_{\mbox{\scriptsize prime } p_i < \gamma} \log_2 n + 
\sum_{\mbox{\scriptsize prime } p_i \in [\gamma, n]} \log_\gamma n,
$$
choosing $\gamma$ as $n/\log n$ gives the bound $\O(n / \log n)$. 
\end{proof}

Then, multiplying the prime factors pairwise until their product
$\binom{n}{k}$ is computed gives the running time shown in
Theorem~\ref{thm:cbc}. 

\begin{theorem}\label{thm:cbc}
A binomial coefficient $\binom{n}{k}$ can be computed in $\O(\M(n))$
time given a prime table from $2$ to $n$. 
\end{theorem}

For Vardi's algorithm, we only care about central binomial
coefficients, but of course these are just a special case of this theorem.

\subsection{Factorial is in $\O(\M(n \log n))$}

Vardi compute the factorial $n!$ by the identity 
\begin{equation}
n! = \bigg( (n)^{r} \binom{n}{n/2} (n/2)!\bigg) (n/2)!, 
\end{equation}
where $r \equiv n \pmod 2$ and $n/2$ denotes integral division, i.e., $n/2 = \lfloor n/2 \rfloor$. Note that there are four terms on the R.H.S. of the identity and each has $\O(n \log n)$ bits. Let $\F(n)$ denote the running time for computing the factorial $n!$. Then, we have the following recurrence relation, 
\begin{equation}
\F(n) = \F(n/2) + \O(\M(n \log n))
\end{equation}
due to Theorems~\ref{thm:main} and~\ref{thm:cbc}. Therefore, we have Theorem~\ref{thm:fac}.

\begin{theorem}\label{thm:fac}
The factorial $n!$ can be computed in $\O(\M(n \log n))$ time.
\end{theorem}

\end{document}